\documentclass[10pt,twocolumn,journal,compsoc]{IEEEtran}
\IEEEoverridecommandlockouts
\usepackage[utf8]{inputenc}
\usepackage[T1]{fontenc}
\usepackage{textcomp}
\usepackage{gensymb}
\usepackage{cite}
\usepackage{graphicx}
\usepackage{url}
\usepackage{amsmath}
\usepackage{latexsym,amssymb,epsfig,color,verbatim}
\usepackage{algorithm}
\usepackage{algpseudocode}
\usepackage{pifont}
\usepackage{slashbox}
\usepackage{subfig}
\usepackage{color}
\usepackage{cancel}
\usepackage{multirow}
\usepackage{amsthm,amsmath}
\usepackage{bm}
\usepackage{tabu}
\usepackage{mathtools}
\usepackage{multirow}
\usepackage{algorithm}
\usepackage{algpseudocode}
\usepackage{mathtools,amssymb,lipsum}

\newcommand{\BEQA}{\begin{eqnarray}}
\newcommand{\EEQA}{\end{eqnarray}}

\newtheorem{lemma}{Lemma}
\newtheorem{theorem}{{\bf Theorem}}
\newtheorem{corollary}{{\bf Corollary}}
\newtheorem{definition}{{\bf Definition}}
\newtheorem{remark}{{\bf Remark}}

\renewcommand\IEEEkeywordsname{Keywords}

\usepackage{cuted}

\begin{document}
	\bstctlcite{IEEEexample:BSTcontrol}
	\title{Resource Sharing in the Edge: A Distributed Bargaining-Theoretic Approach}
\author{Faheem~Zafari,~\IEEEmembership{Student Member,~IEEE,}
	Prithwish~Basu,~\IEEEmembership{Senior Member,~IEEE,}\\
	Kin~K. Leung,~\IEEEmembership{Fellow,~IEEE,} 
	Jian~Li,~\IEEEmembership{Member,~IEEE,}
	Don~Towsley,~\IEEEmembership{Fellow,~IEEE,}
	and~Ananthram~Swami,~\IEEEmembership{Fellow,~IEEE, }	\thanks{}

\IEEEcompsocitemizethanks{\IEEEcompsocthanksitem Faheem Zafari and Kin K. Leung are  with the Department
	of Electrical and Electronics Engineering, Imperial College London, London,
	UK.   \protect\\
	E-mail:  \{faheem16,kin.leung\}@imperial.ac.uk
	\IEEEcompsocthanksitem  Don Towsley is with College of Information and Computer Sciences, University of Massachusetts Amherst, Amherst, MA 01003, USA.\protect\\
	E-mail:  towsley@cs.umass.edu
	\IEEEcompsocthanksitem Prithwish Basu is with BBN Technologies,  MA, USA.
	\protect\\
	Email: prithwish.basu@raytheon.com
	\IEEEcompsocthanksitem Ananthram Swami is with the U.S. Army Research Laboratory, Adelphi, MD 20783, USA.
	\protect\\
	Email: ananthram.swami.civ@mail.mil
	\IEEEcompsocthanksitem Jian Li is with  is with the Department of Electrical and Computer Engineering, Binghamton University, the State University of New York Binghamton, NY 13902, USA.
	\protect\\
	Email: lij@binghamton.edu
	
}}
\IEEEtitleabstractindextext{

\begin{abstract}
The growing demand for edge computing resources, particularly due to increasing popularity of Internet of Things (IoT), and  distributed machine/deep learning applications poses a significant challenge.  On the one hand, certain edge service providers (ESPs) may not have sufficient resources to satisfy their applications according to the associated service-level agreements. On the other hand, some ESPs may have additional unused resources.  In this paper, we propose a resource-sharing framework that allows different ESPs to optimally utilize their resources and improve the satisfaction level of applications subject to constraints such as communication cost for sharing resources across ESPs. Our framework considers that different ESPs have their own objectives for utilizing their resources, thus resulting in a multi-objective optimization problem.  We present an $N$-person \emph{Nash Bargaining Solution} (NBS) for resource allocation and sharing among ESPs with \emph{Pareto} optimality guarantee. Furthermore, we propose 
a \emph{distributed}, primal-dual algorithm to obtain the NBS by proving that the strong-duality property holds for the resultant resource sharing optimization problem.
 Using synthetic and real-world data traces, we show numerically that the proposed NBS based framework not only enhances the ability to satisfy applications' resource demands, but also improves utilities of different ESPs. 
\end{abstract}}
	\maketitle

\section{Introduction}\label{sec:intro}
\subsection{Motivation}
Edge computing has received much attention recently as it enables on-demand provisioning of computing resources for different applications and tasks at the network edge \cite{shi2014online,shi2016edge,satyanarayanan2017emergence}. One of the fundamental advantages of edge computing is that it can provide resources 
with low latencies when compared with traditional cloud computing architecture \cite{shi2016edge}. The demand for edge computing has further increased due to the advent of Internet of Things (IoT) \cite{atzori2010internet} and wide-scale use of machine and deep learning \cite{li2018learning} in different industries as  these learning-based models can  be trained and run using edge computing nodes with adequate storage and computing power \cite{wang2018edge}. 

\par A typical edge computing system consists of a large number of edge nodes  that have different types of resources. Edge service provider (ESP) earns a \emph{utility} for allocating resources to different applications and guarantees to provide  resources according to a \emph{Service Level Agreement} (SLA). However, when compared with cloud computing systems and data centers, resources in an edge setting are limited. Therefore, optimal use and allocation of these limited resources has been an active area of research. 
Even when resources are available, allocating  them to  applications with the goal of maximizing overall ESP \emph{utility} is a difficult problem. Furthermore, the aforementioned resource intensive paradigms such as deep learning, and data analytics  exacerbate the problem by challenging the scalability of traditional resource allocation techniques. \par 

ESPs typically provide enough resources to different applications at their edge nodes to  meet the peak demand. However, it is highly likely that  resources of one ESP will be over-utilized, while other ESP's resources will be under-utilized.  
For example, an ESP provisions resources to an application at an edge node that is physically closest to the requesting application. However, if the closest edge node has a resource deficit or is overloaded, the request can be satisfied through ESP's next closest edge node that  may physically  be at a distant location or deep in the network such as at the data center. This incurs high cost and  causes high latency that may not be acceptable for delay constrained applications. One possible solution is  to create a shared  resource pool with other ESPs that are  physically closer  \cite{he2018s,openedge,openfog,etsi}. Such a resource pool allows  ESPs to share and use  resources whenever needed to meet their dynamic demands. This cooperation and resource sharing among ESPs seem beneficial for them, because it is unlikely that resources of different ESPs will be   simultaneously over-utilized. Furthermore, cooperation among different service providers or vendors also exists in real life as well. For example, Amazon and Netflix compete with each other in video streaming business. However, Netflix also relies on Amazon's web services to provide its video streaming services. 
\par  
Our work has also been  motivated by military settings in which two or more  \emph{coalition partners}\footnote{Each coalition partner can be considered an ESP. } are jointly conducting a military operation using different resources from the partners. As military settings usually have strict  latency and reliability requirements, resources are placed at the edge to fulfill the aforementioned requirements. 
Hence, coalition settings can be considered a practical scenario for edge computing. In contrast with commercial edge computing settings, military settings require a higher reliability and robustness to ensure timely availability of resources. 
As seen in Figure \ref{fig:system_model}, both coalition partners $2$ and $3$ satisfy all their applications and have resource surpluses whereas partner $1$ has resource deficits when working alone. However, through cooperation among these coalition partner, $1$ satisfies its applications by using resources of other coalition partners resulting in  an improved utility for all partners. It is evident that if  different coalition partners do not share resources, different application requests cannot be satisfied. Therefore, there is a need for a framework that allows resource sharing among these coalition partners (i.e., ESPs). 
Furthermore, the framework needs to be distributed as a centralized solution may not be acceptable to the different coalition partners or ESPs because:
\begin{itemize}
	\item An adversary can target the central system  causing the entire military operation to fail.
	\item Coalition partners need to reach a consensus  for choosing the \emph{central} node that runs the resource sharing algorithm. 
	\item A centralized framework requires  a large amount of information to be transmitted to the central node. This may not be feasible or preferred by the coalition partners as certain information may be private.  
\end{itemize} 
Such a distributed resource sharing gives rise to a number of questions such as: 
 \begin{enumerate}
 	\item \emph{Should an ESP help another ESP by sharing resources?}
 	\item \emph{How should  resources be allocated to applications across different ESPs, while considering issues such as  communication cost and resource fragmentation? }
 	\item \emph{How can  ESPs share the profits of resource sharing?}
 \end{enumerate}
We  answer these questions in this paper. 
\subsection{Methodology and contributions}
In this paper, we consider a number of ESPs that share their resources and form a  resource pool to satisfy resource requests of different applications as shown in Figure \ref{fig:system_model}.
We formulate such resource sharing and allocation among ESPs as a   multi-objective optimization (MOO) problem, for which our goal is to achieve a Pareto optimal solution.  Furthermore, since Pareto solutions are spread over the \emph{Pareto frontier} \cite{cho2017survey}, choosing a single solution among them is challenging. As the \emph{Nash Bargaining Solution}  (NBS)\cite{nash1950bargaining} guarantees the provision of  a fair and Pareto optimal solution to such a MOO problem,  we develop a distributed NBS-based framework for  resource allocation and sharing here.  \emph{To the best of our knowledge, this is the first generic NBS-based framework for resource sharing and allocation among ESPs. } 

 \begin{figure}
	\centering
	\includegraphics[width=0.49\textwidth]{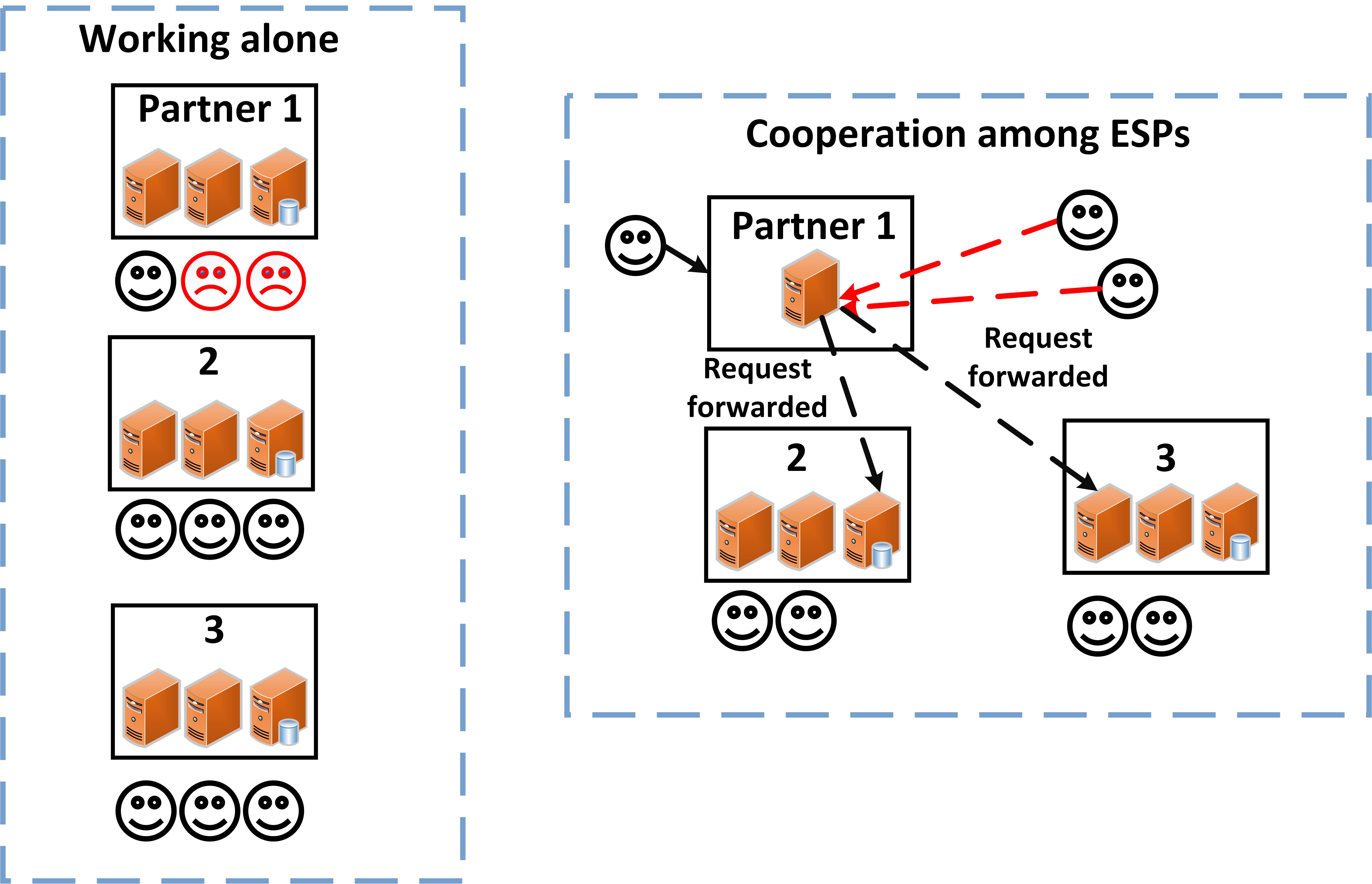}
	\vspace{0.05in}
	\caption{Cooperation among ESPs}
	\protect\label{fig:system_model}
\end{figure}
The main contributions of this paper are: 
\begin{enumerate}
  
	\item We present an NBS based  resource  sharing framework for ESPs with different objectives in allocating resources to meet the dynamic demands. Our framework also considers practical engineering constraints such as communication costs and resource fragmentation. 
\item We show that ESPs can benefit from sharing resources  with other ESPs as they can earn a higher profit and improve the average application satisfaction.    
\item We  show that  \emph{strong-duality} property holds for the formulated problem, which enables us to propose a distributed  algorithm to obtain the NBS.   
\item 	We evaluate the performance of our algorithms using  synthetic and real world traces.  Results show that resource sharing improves the utilities of ESPs, increases resource utilization and enhances average user (application) satisfaction. Furthermore, our results show that the profit sharing among  ESPs is also fair. 
\end{enumerate}
The rest of this paper is structured as follows. We present the system model in Section~\ref{sec:sysmodel}. 
In Section~\ref{sec:nbs}, we first present a primer on NBS. We then describe our proposed  NBS framework  for resource sharing and allocation. In Section~\ref{sec:dist}, we present a distributed algorithm for obtaining the NBS. In Section~\ref{sec:exp_results}, we present simulation results of the proposed framework for different settings. We describe the related work in Section \ref{sec:related},  
and conclude the paper in Section~\ref{sec:conclusion}.


\section{System Model}\label{sec:sysmodel}
Let $\mathcal{N}=\{1,2, \cdots, N\}$ be the set of all ESPs. 
We assume that each  provider has a set of $\mathcal{K}=\{1,2,\cdots, K\}$  types of resources such as communication, computation and storage resources.  
$C_{n} = \{ C_{n,1},\cdots, C_{n,K}\}$, with  $C_{n,k}$ denoting the amount of type $k$ resources available at service provider $n$. 

  Each service provider $n$ has a set of  native applications $\mathcal{M}_n= \{1,2,\cdots, M_n\}$. 
The set of all applications that request resources from all service providers  is given by $\mathcal{M}=\mathcal{M}_1\cup\mathcal{M}_2\cdots\cup \mathcal{M}_N , \;$ where we assume  $\mathcal{M}_i \cap \mathcal{M}_j=\emptyset, \; \forall i \neq j,$ i.e., each application initially  demands resources from only  its native service provider. We also define $\overline{\mathcal{M}}$ to represent  $\mathcal{{M}}\backslash \mathcal{{M}}_n$ and $\overline{\mathcal{N}}$ to represent $\mathcal{{N}}\backslash n $.  
Every ESP  $n \in \mathcal{N}$ has a request (requirement) matrix $R_{n}$,
\begin{equation}
\label{eq:R_Req}
R_{n}=\Biggl[\begin{smallmatrix}
\mathbf{r_{n}^1}\\ 
.\\ 
.\\ 
.\\
\mathbf{r_{n}^{{M}_n}}
\end{smallmatrix}\Biggr] = \Biggl[\begin{smallmatrix}
r_{n,1}^{1} & \cdots  &\cdots   & r^{1}_{n,{K}} \\ 
. & . &.  &. \\ 
. & . &.  &. \\
. & . &.  &. \\
r_{n,1}^{{M}_n}&\cdots  &\cdots   & r_{n,K}^{{M}_n}
\end{smallmatrix}\Biggr],
\end{equation}
Here  $r_{n,k}^{j}$ is the amount of resource $k$ that application $j \in \mathcal{M}_n$ requires. When an ESP is working alone (no sharing of  its resources with any other service provider),  its objective is to maximize its utility by allocating  resources to its native applications. 
 A service provider $n$ earns a  utility $u^j (x_{n,k}^j)$ by allocating $x_{n,k}^j$ amount of resource $k$ to application $j\in \mathcal{M}_n$, where the vector $\mathbf{x}_n^j= [x_{n,1}^j, x_{n,2}^j,\cdots, x_{n,K}^j]^T$. 
 Table \ref{tab:not} contains notations used throughout the paper. We present the optimization formulation for a single service provider in Section \ref{sec:prob}, followed by a formulation for the multiple service provider problem in Section \ref{sec:multP}.  


\begin{table}[]
	\centering
	\caption{List of notations used throughout the paper}
	\begin{tabular}{|p{1cm}|p{6.9cm}|}
		\hline
		\textbf{Notation}	& \textbf{Description} \\ \hline
		$\mathcal{{N}}, N, n$	& Set, number and index of ESPs  \\ \hline
		$\mathcal{K}, K, k$	& Set, number and index of resources \\ \hline
		$\mathcal{M}, M, j$	& Set, number and index of applications \\ \hline
		$\mathcal{M}_{n}$	& Set of native applications at ESP $n$  \\ \hline
		$C$	& Capacity vector of all ESPs \\ \hline
		$C_n$	& Capacity vector of ESP $n$ \\ \hline
		$C_{n,k}$	& Capacity of resource $k$ at ESP $n$ \\ \hline
		$R_{n}$	& Request matrix at ESP $n$ \\ \hline
		$r_{n,k}^{j}$	& Request of application $j$ for resource $k$ from ESP $n$ \\ \hline
		$x_{n,k}^{j}$	& Allocation decision of resource $k$ for application $j$ at ESP $n$ \\ \hline
	$\mathbf{x}_{n}^{j}$	& Allocation decision vector for application $j$ at ESP $n$ when working alone, i.e., $\mathbf{x}_{n}^{j}=[x_{n,1}^j, \cdots, x_{n,K}^j]^T$  \\ \hline

	$\mathbf{X}_n$ & Allocation decision for ESP $n$ in the  resource sharing case \\ \hline 
	$\mathbf{X}$ & Allocation decision for the entire set of ESPs  \\ \hline 
			$u_{}^j({x}_{n,k}^j)$	& Utility ESP $n$ earns  by allocating resource $k$ to application $j$ \\ \hline
		$u_{}^j(\mathbf{x}_{n}^j)$	& Utility ESP  $n$ earns  by allocating vector of  resources to application $j$ \\ \hline
	\end{tabular}
	\label{tab:not}
\end{table}

\subsection{Problem Formulation for Single Service Provider } \label{sec:prob}
 We  first present the resource allocation problem for a stand-alone single service provider (i.e., no resource sharing with other service providers). 
For a single ESP $n \in \mathcal{N}$, the allocation decision  consists of  vectors $\mathbf{x}_n^1, \cdots, \mathbf{x}_n^{M_n}$. The optimization problem is: 
\begin{subequations}\label{eq:opt_single}
	\begin{align}
	\centering
\max_{\mathbf{x}_n^1, \cdots, \mathbf{x}_n^{M_n}}\quad& \sum_{\substack{j \in \mathcal{M}_n,\\k \in \mathcal{K}}}u^j({x}_{n,k}^j), \label{eq:objsingle}\\
	\text{s.t.}\quad & \sum_{j} x_{n,k}^{j}\leq C_{n,k}, \quad \forall k \in \mathcal{K}, \label{eq:singlefirst} \displaybreak[0]\\
	& x_{n,k}^{j} \leq r^{j}_{n,k}, \quad \forall\; j\in \mathcal{M}_n, k \in \mathcal{K}, \label{eq:singlesecond} \displaybreak[1]\\
	&  x_{n,k}^{j} \geq 0, \quad \forall\; j\in \mathcal{M}_n, k \in \mathcal{K}. \label{eq:singlethird} \displaybreak[2]
	\end{align}
\end{subequations}
The goal of a single service provider in solving this single objective optimization (SOO) problem, as mentioned earlier,  is   to  maximize its utility by appropriately allocating  resources. 
The first constraint  \eqref{eq:singlefirst} indicates that  allocated resources cannot exceed capacity. The second constraint   \eqref{eq:singlesecond} reflects that allocated resources should not exceed the requested amounts.  The last constraint, \eqref{eq:singlethird}  says the allocation  cannot be negative. However, it is possible that a service provider $n$ may 
earn a larger utility by providing its resources to applications of other service providers or it may not have sufficient resources to satisfy  requests of all its native applications. On the other hand, there may be another service provider $m\in \mathcal{{N}}\backslash n$ that may have a surplus of resources, which can be ``rented" by service provider $n$. 
Below, we discuss  resource sharing among these service providers.

\subsection{Multiple Service Providers Problem Formulation}\label{sec:multP}
Allowing resource sharing among  service providers, while considering their objectives, could improve resource utilization and  application satisfaction. 
Let $D^j(x_{n,k}^j)$ denote the communication cost of serving application $j \in  \mathcal{{M}}_m$ at service provider $n$ rather than at its native service provider $m$. 
By sharing its resources, ESP $n$ earns a utility,
\begin{equation}\label{eq:util}
u^j (\sum_{m \in \overline{\mathcal{N}}}{x}_{m,k}^j+x_{n,k}^j)-u^j(\sum_{m \in \overline{\mathcal{N}}}{x}_{m,k}^j)-D^j(x_{n,k}^j),
\end{equation}
after allocating  $x_{n,k}^j$ for application $j\in \mathcal{{M}}_m$, i.e., the net utility ESP $n$ earns is calculated as the \emph{differential utility} ($u^j (\sum_{m \in \overline{\mathcal{N}}}{x}_{m,k}^j+x_{n,k}^j)-u^j(\sum_{m \in \overline{\mathcal{N}}}{x}_{m,k}^j)$) earned due to providing  $x_{n,k}^k$ amount of resources minus the communication cost between the native and non-native ESP. 
  We assume that  $ u^j (x_{n,k}^j)=u^j (x_{m,k}^j)$ when $x_{n,k}^j=x_{m,k}^j$. 
   We also let
   \begin{equation*}
 \mathbf{x}^j= [\sum_{n \in \mathcal{{N}}}x_{n,1}^j, \sum_{n \in \mathcal{{N}}}x_{n,2}^j,\cdots, \sum_{n \in \mathcal{{N}}}x_{n,K}^j]^T,
   \end{equation*}
 i.e., the total resource allocated to any application $j \in \mathcal{M}$ is the sum of resources allocated to application $j$ from all service providers.  The resource sharing and allocation algorithm, based on resource requests and  capacities of service providers, has to make an allocation that optimizes  utilities of all  service providers $n \in \mathcal{N}$ and satisfy user requests as well.  The allocation decision is given by
 $\mathbf{X}=\{\mathbf{X}_1, \mathbf{X}_2, \cdots, \mathbf{X}_N\}$, where $\mathbf{X}_n,\; \forall n \in \mathcal{N}$ is given by: 

\begin{equation}
\label{eq:X}
\mathbf{X}_{n}=\Biggl[\begin{smallmatrix}
\mathbf{x}^{1}_n\\ 
.\\ 
.\\ 
.\\
\mathbf{x}^{|\mathcal{{M}}|}_n
\end{smallmatrix}\Biggr] = \Biggl[\begin{smallmatrix}
x_{n,1}^{1} & \cdots  &\cdots   & x^{1}_{n,{K}} \\ 
. & . &.  &. \\ 
. & . &.  &. \\
. & . &.  &. \\
x_{n,1}^{|\mathcal{{M}}|}&\cdots  &\cdots   & x_{n,K}^{|\mathcal{{M}}|}
\end{smallmatrix}\Biggr].
\end{equation}
Each service provider aims to  maximize the sum of utilities by  allocating its resources to its native applications, and  allocating resources  to  applications belonging to other service providers. 
Each provider $n \in \mathcal{{N}}$ solves the following multi-objective optimization problem.

\begin{subequations}\label{eq:opt_higher}
	\begin{align}
	\max_{\mathbf{X}_n}\quad& \sum_{\substack{j \in \mathcal{M}_n,\\k \in \mathcal{K}}}\Big(u^j(\sum_{m \in \overline{\mathcal{N}}}{x}_{m,k}^j+{x}_{n,k}^j)- u^j(\sum_{m \in \overline{\mathcal{N}}}{x}_{m,k}^j) \Big) \nonumber \\
	&+\Big(\sum_{\substack{l \in \overline{\mathcal{M}},\\k \in \mathcal{K}}}\big( u^l(\sum_{m \in \overline{\mathcal{N}}}{x}_{m,k}^l+{x}_{n,k}^l)-u^l(\sum_{m \in \overline{\mathcal{N}}}{x}_{m,k}^l)\nonumber\\
	&- D^l({x}_{n,k}^l)\big)\Big), \label{eq:obj}\\ 
	\text{s.t.}\quad 
& \sum_{j} x_{n,k}^{j}\leq C_{n,k}, \quad \forall\;  k \in \mathcal{K}, n \in \mathcal{N}, \label{eq:obj1} \displaybreak[0]\\
	& \sum_{m \in \mathcal{N}}x_{m,k}^{j} \leq r_{n,k}^{j}, \quad \forall\; j\in \mathcal{M}, k \in \mathcal{K}, n \in \mathcal{N},  \label{eq:obj2}\displaybreak[1]\\
	&  x_{n,k}^{j} \geq 0, \quad \forall\; j\in \mathcal{M}, k \in \mathcal{K}, n \in \mathcal{N}, \label{eq:obj3} \displaybreak[2]\\
	& u^l(\sum_{m \in \overline{\mathcal{N}}}{x}_{m,k}^l+{x}_{n,k}^l)-
	u^l(\sum_{m \in \overline{\mathcal{N}}}{x}_{m,k}^l)- \nonumber \\
	& D^l({x}_{n,k}^l)\geq0, \forall l \in \mathcal{M}\backslash\mathcal{M}_n, k \in \mathcal{K}, n \in \mathcal{N}\label{eq:obj4} .
	\end{align}
\end{subequations}
The first summation term in \eqref{eq:obj} represents the utility earned by an ESP providing resources to the native applications, whereas the second summation term describes the utility earned by providing resources to non-native applications. 
 Note that constraint \eqref{eq:obj1} indicates that the total allocated resources cannot exceed the resource capacity of the service providers. \eqref{eq:obj2} states that the total amount of resources allocated to any application using the resource-sharing framework cannot exceed the amount of  requested resources. \eqref{eq:obj3}  says that the resource allocation cannot be negative  whereas \eqref{eq:obj4} indicates that the incremental increase in utility earned by providing resources to non-native applications should be non-negative. 
 
 \subsection{Assumptions}\label{sec:assumption}
 In our model, we assume that  each utility is a  concave injective function  for which the inverse of the first derivative exists, such as $(1-e^{-x})$. 
Strictly speaking,  our centralized NBS framework requires only concave injective utility functions\cite{yaiche2000game}. However,  the existence of the inverse of the first derivative is required for the distributed NBS (see details in Section \ref{sec:nbs}).  The communication cost is a convex function, hence the objective function in \eqref{eq:opt_higher} is  concave. 
All resources are fully utilized in the  optimal solution.  However, there are enough resources to provide a positive utility to all  service providers when sharing resources.  This is a realistic assumption as the demand for resources is usually more than the supply.  
 Furthermore, we assume that when ESPs share resources,  there exist solutions that are better than when they are all  working alone. This assumption can be relaxed, that is, if certain ESPs cannot improve their utility using the bargaining solution, they will not  participate in the resource sharing framework. However, the framework can still be used for the remaining  ESPs.

 \subsection{Choice of  utility function}
  While our framework works with any utility that satisfies the conditions in Section \ref{sec:assumption}, choosing a suitable utility function along with the communication cost can minimize \emph{resource fragmentation}\footnote{We define resource fragmentation as the process in which  resources provided to an application $j$ are split across multiple ESPs rather than a single ESP. }. 
 Generally, a concave utility has a steeper slope at start that becomes flatter as more resources are allocated, i.e., the rate of increase in the payoff for allocating  resources reduces with increase in the amount of allocated resources.  This results in  an ESP providing a fraction of originally requested resources to an application and keeping the remaining resources for other applications. Hence, the application  does not get all the resources it needed, and has to ask another ESP for more resources, resulting in resource fragmentation. 
 For example, assume that two applications require $4$ units of a particular resource  from ESP $n$ that only has $3$ units available. Due to the nature of many concave utilities, ESP $n$ for maximizing its utility will provide part of its resources to one application and the remaining resources to the other application rather than providing all $3$ available units to one application and borrowing resources for the other application. This causes resource fragmentation, i.e., both  applications received only a part of the required resources and they will need to obtain the remaining amount from other ESP(s). To avoid such problems, we propose that the utility function should consider:
  \begin{itemize}
  	\item Minimum acceptable amount of resources: The ESPs should earn either zero utility or a negative utility if the resources provided are not within $\delta$ units of the  requested resource $r_{n,k}^j$. Rather than using a utility function such as  $1-e^{-(x_{n,k}^j)}$ that pays  the ESP even when a small amount of resources are provided, it is better to use $1-e^{-(x_{n,k}^j-r_{n,k}^j +\delta)}$  that becomes positive only when the allocation $x_{n,k}^j$ is within $\delta_t$ units of the requested resource $r_{n,k}^j$. Such a utility along  with the communication cost helps minimize the aforementioned fragmentation problem. 
  \end{itemize}


\section{NBS for Resource Sharing among CSPs}\label{sec:nbs}
We first present an introduction to NBS and then discuss our proposed NBS based resource sharing framework. 
\subsection{Primer on Nash Bargaining Solution (NBS)}\label{sec:prelim}
We use a two-player game as a toy example to introduce NBS.  Consider two players $1,$ and  $2$ that need to reach an agreement (e.g., resource allocation decision) in an outcome space $\mathcal{A}\subseteq \mathbb{R}^2$ \cite{han2012game}. Both players have a utility given by $u_1$ and $u_2$ over the space $\mathcal{A} \cup \{D\}$ where $D$ specifies a \emph{disagreement outcome} for players in the event of a disagreement, i.e., when  two players cannot reach an agreement. Let $\mathcal{S}$ be the set of all possible utilities that both  players can achieve: 

\begin{align}
\label{eq:possibleut}
\mathcal{S}= \{(u_1(a_1),u_2(a_2)) | (a_1,a_2) \in \mathcal{A} \}
\end{align}

We also define $d=(d_1, d_2)$, where $d_1=u_1(D)$ and $d_2=u_2(D)$, as the payoff each player receives at the disagreement point. We define the bargaining problem  as the pair $(\mathcal{S},d)$ where $\mathcal{S} \subset \mathbb{R}^2$ and $d \in \mathcal{S}$ such that 
\begin{itemize}
	\item $\mathcal{S}$ is a convex and compact set;
	\item There exists $s \in \mathcal{S}$ such that $s > d$.
\end{itemize}
In NBS, the goal is to obtain a function $f(\mathcal{S},d)$ that provides  a unique outcome in $ \mathcal{S}$ for every bargaining problem $(\mathcal{S},d)$. Nash studied the possible outcomes (agreements) that players can reach whereas the agreements, along with the Pareto optimality, must also  satisfy the following set of  axioms (also called \emph{fairness} axioms) \cite{han2012game}:
\begin{enumerate}
	\item \textbf{Symmetry:}  The bargaining solution will not discriminate among  players if players are indistinguishable, i.e., players have identical utilities. 
	\item \textbf{Invariance to equivalent utility representation:} If a bargaining problem  $(\mathcal{S},d)$ is transformed into another bargaining problem  $(\mathcal{S}',d')$ where $s_i'=\gamma_is_i+\zeta_i$ and $d_i'=\gamma_id_i+\zeta_i$, $\gamma_i >0$, then $f( \mathcal{S}',d')=\gamma_if(\mathcal{S},d)+\zeta_i$
	\item \textbf{Independence of irrelevant alternatives:} For any two bargaining problems  $(\mathcal{S},d)$ and  $(\mathcal{S}',d)$ where $\mathcal{S}'\subseteq \mathcal{S}$, if $f(\mathcal{S},d) \in \mathcal{S}'$, then $f(\mathcal{S}',d)=f(\mathcal{S},d)$. 
\end{enumerate}
\cite{nash1950bargaining} shows that there is a unique bargaining solution that satisfies above axioms. We present it in the following theorem. 
\begin{theorem}{\cite{han2012game}}
	There exists a unique solution satisfying the aforementioned axioms and this solution is the pair of utilities $(s_1^*,s_2^*) \in \mathcal{S}$ that solves the following optimization problem:
	\begin{align}
	\label{eq:nashProb}
	\max_{s_1,s_2}\;(s_1-d_1)(s_2-d_2),\; s.t. (s_1,s_2)\in \mathcal{S}, \; (s_1,s_2)\geq (d_1,d_2).
	\end{align}\
\end{theorem}
The solution of \eqref{eq:nashProb} is the NBS. 
The above framework can be extended to $N$ players  by allowing $\mathcal{S}$ to be an $N$-dimensional space \cite{harsanyi1963simplified}. For this case, the bargaining problem $(\mathcal{S},d)$, with $d=(d_1,d_2,\cdots, d_N)$ as the disagreement point, becomes the unique solution of the optimization problem below.
	\begin{align}\vspace{-0.1in}
\label{eq:nashProbNplayer}
\max_{s_1,\cdots,s_N}&\;\prod_{n=1}^{N}(s_n-d_n),\; \nonumber \\
s.t.\quad  &(s_1,\cdots,s_N)\in \mathcal{S}, \nonumber \\
& (s_1,\cdots,s_N)\geq (d_1,\cdots, d_N).
\end{align}
Solving \eqref{eq:nashProb} is easier compared to the $N-$player bargaining problem in \eqref{eq:nashProbNplayer}\cite{han2012game}. 
In this paper, we transform our problem into an equivalent convex problem that is comparatively easier to solve.  
 
\begin{remark}
  As each player in an  $N$-player bargaining game has a particular objective to optimize, the resulting problem is multi-objective,  where the goal is to obtain a Pareto optimal solution. NBS is a fair and Pareto optimal solution for such MOO problems, provided that the fairness axioms are satisfied.  
\end{remark}

\subsection{Proposed Framework}
As mentioned earlier, NBS is  a Pareto optimal and fair solution in settings that involve different players (ESPs in our case) where each player has an objective to optimize. Therefore, it can be used to solve our MOO problem in \eqref{eq:opt_higher} as we have different ESPs  that need to optimize their objectives and improve their utilities over what they would receive  working alone. 
 We first specify the disagreement point for our bargaining problem. If the ESPs  cannot come to an agreement, they can all start working alone.  Hence the disagreement point  is the solution to the SOO problem given in \eqref{eq:opt_single} for all ESPs. Let us denote the solution to the SOO problem by $d_n^0, \forall\; n \in \mathcal{N}$.   In the cooperative setting\footnote{Service providers share resources among each other.}, the utility (represented by $s_n$ in \eqref{eq:nashProbNplayer})  for an ESP $n \in \mathcal{N}$  is given by: 
 \begin{align}
U= &\sum_{\substack{j \in \mathcal{M}_n,\\k \in \mathcal{K}}}\Big(u^j(\sum_{m \in \overline{\mathcal{N}}}{x}_{m,k}^j+{x}_{n,k}^j)- u^j(\sum_{m \in \overline{\mathcal{N}}}{x}_{m,k}^j) \Big) \nonumber \displaybreak[0] \\
&\hspace*{-0.2in}+\Big(\sum_{\substack{l \in \overline{\mathcal{M}},\\k \in \mathcal{K}}}\big( u^l(\sum_{m \in \overline{\mathcal{N}}}{x}_{m,k}^l+{x}_{n,k}^l)-u^l(\sum_{m \in \overline{\mathcal{N}}}{x}_{m,k}^l)- \displaybreak[1]\nonumber\\
&\hspace*{-0.1in} D^l({x}_{n,k}^l)\big)\Big). \nonumber 
 \end{align} 
  Below, we  present the centralized NBS algorithm. 
 \subsubsection{Centralized NBS}
 We first present the optimization problem to obtain NBS for our $N-$ESP bargaining game  \cite{harsanyi1963simplified} and then present its equivalent problem \cite{yaiche2000game} that is computationally efficient to solve. 
 \begin{theorem} The NBS  for the MOO optimization problem in \eqref{eq:opt_higher} can be  obtained by solving the following optimization problem:
 \begin{subequations}\label{eq:nashProbNplayermod}
 	\begin{align}
  	\max_{\mathbf{X}}\;& \prod_{n=1}^{N}\hspace*{-0.05in}\bigg(\sum_{\substack{j \in \mathcal{M}_n,\\k \in \mathcal{K}}}\Big(u^j(\sum_{m \in \overline{\mathcal{N}}}{x}_{m,k}^j+{x}_{n,k}^j)-u^j(\sum_{m \in \overline{\mathcal{N}}}{x}_{m,k}^j) \Big)  \nonumber \\
  	&+\Big(\sum_{\substack{l \in \overline{\mathcal{M}},\\k \in \mathcal{K}}}\big( u^l(\sum_{m \in \overline{\mathcal{N}}}{x}_{m,k}^l+{x}_{n,k}^l)-u^l(\sum_{m \in \overline{\mathcal{N}}}{x}_{m,k}^l)\nonumber\\
  	&- D^l({x}_{n,k}^l)\big)\Big)-d_n^0
  	\bigg),\label{eq:objmain} \\
 	\text{s.t.}\quad & Constraints \; in \; \eqref{eq:obj1}-\eqref{eq:obj4}, \nonumber \\
 	&\sum_{\substack{j \in \mathcal{M}_n,\\k \in \mathcal{K}}}\Big(u^j(\sum_{m \in \overline{\mathcal{N}}}{x}_{m,k}^j+{x}_{n,k}^j)-u^j(\sum_{m \in \overline{\mathcal{N}}}{x}_{m,k}^j) \Big)  \nonumber \\
 	&+\Big(\sum_{\substack{l \in \overline{\mathcal{M}},\\k \in \mathcal{K}}}\big( u^l(\sum_{m \in \overline{\mathcal{N}}}{x}_{m,k}^l+{x}_{n,k}^l)-u^l(\sum_{m \in \overline{\mathcal{N}}}{x}_{m,k}^l)\nonumber\\
 	&- D^l({x}_{n,k}^l)\big)\Big)> d_n^0, \quad \forall \; n \in \mathcal{N}.\label{eq:disagreemt}
 	\end{align}
 \end{subequations}
 \end{theorem}
\begin{proof}
The feasible set for the above optimization problem is convex and compact, since all constraints are convex and intersection of convex sets is a convex set. Furthermore, based on our assumptions, there exist solutions (allocation and sharing decisions) that  provide better utility than the disagreement point $d_n^0, \forall\; n \in \mathcal{N}$. Hence, the solution of the optimization problem in \eqref{eq:nashProbNplayermod} is the NBS.
\end{proof}
However, solving \eqref{eq:nashProbNplayermod} is computationally complex and the problem is not always convex.  Therefore, there is a need for an efficient method to obtain the NBS. 
Toward this goal,  we transform the problem \eqref{eq:nashProbNplayermod} into an equivalent problem as proposed in \cite{yaiche2000game}. 

\begin{corollary}
The NBS for \eqref{eq:opt_higher} is obtained by solving the following optimization problem:
	\begin{align}\label{eq:objmain2} 
	\max_{\mathbf{X}}\;& \sum_{n=1}^{N}\ln\bigg(\sum_{\substack{j \in \mathcal{M}_n,\\k \in \mathcal{K}}}\Big(u^j(\sum_{m \in \overline{\mathcal{N}}}{x}_{m,k}^j+{x}_{n,k}^j)- \nonumber \\
	&u^j(\sum_{m \in \overline{\mathcal{N}}}{x}_{m,k}^j) \Big) +\Big(\sum_{\substack{l \in \overline{\mathcal{M}},\\k \in \mathcal{K}}}\big( u^l(\sum_{m \in \overline{\mathcal{N}}}{x}_{m,k}^l+{x}_{n,k}^l)\nonumber\\
	&-u^l(\sum_{m \in \overline{\mathcal{N}}}{x}_{m,k}^l)- D^l({x}_{n,k}^l)\big)\Big)-d_n^0
	\bigg),\\
	\text{s.t.}\quad &  Constraints \; in \; \eqref{eq:obj1}-\eqref{eq:obj4}, \; \eqref{eq:disagreemt}. \nonumber 
	\end{align}
\end{corollary}
\begin{proof}
	ESP utilities are concave and  bounded above. Furthermore, the feasible set along with the set of achievable utilities ($\mathcal{S}$) is convex and compact (due to the nature of  utilities and constraints). 
	Utilities of ESPs are also injective functions of the allocation decision. Hence \eqref{eq:objmain2} is equivalent to \eqref{eq:nashProbNplayermod} \cite{yaiche2000game}. Since   the logarithm of any positive real number is concave\cite{boyd2004convex},  \eqref{eq:objmain2} is a convex optimization problem with a unique solution, which is the NBS. 
\end{proof}
The centralized Algorithm \ref{algo:centralnbs} provides the allocation decision using NBS. 

\begin{algorithm}
	\begin{algorithmic}[]
		\State \textbf{Input}:  $C, R$, and vector of utility functions of all ESPs $\boldsymbol{u}$ 
		\State \textbf{Output}: The optimal resource allocation $\mathbf{X}$ and payoffs of all ESPs
		\State \textbf{Step $1$:}
		\For{\texttt{$n \in \mathcal{N}$}}
		\State $ {d_n^0} \leftarrow$ \texttt{Objective function at optimal point in Equation \eqref{eq:opt_single}}
		\EndFor
		\State \textbf{Step $2$:} $\mathbf{X} \leftarrow$ \texttt{Solution of the optimization problem in Equation \eqref{eq:objmain2}} 
	\end{algorithmic}
	\caption{Centralized Algorithm for NBS}
	\label{algo:centralnbs}
\end{algorithm}

	\begin{figure*}[!t]
	\normalsize
	\begin{align}
	\label{eq:lagrangian}
	\mathcal{L}(\mathbf{X}, \boldsymbol{\alpha}, \boldsymbol{\beta, \zeta,\gamma,\pi})&= 
	\sum_{n=1}^{N}\ln\bigg(\sum_{\substack{j \in \mathcal{M}_n,\\k \in \mathcal{K}}}\Big(u^j(\sum_{m \in \overline{\mathcal{N}}}{x}_{m,k}^j+{x}_{n,k}^j)- u^j(\sum_{m \in \overline{\mathcal{N}}}{x}_{m,k}^j) \Big)+\Big(\sum_{\substack{l \in \overline{\mathcal{M}},\\k \in \mathcal{K}}}\big( u^l(\sum_{m \in \overline{\mathcal{N}}}{x}_{m,k}^l+{x}_{n,k}^l)-\nonumber\\
	&\hspace*{-1.3in}\quad u^l(\sum_{m \in \overline{\mathcal{N}}}{x}_{m,k}^l)- D^l({x}_{n,k}^l)\big)\Big)-d^0_n\bigg) +   \quad \sum_{\substack{n \in \mathcal{N},\\j \in \mathcal{M},\\k \in \mathcal{K}}}\gamma_{n,k}^{j}x_{n,k}^{j} +\sum_{\substack{n \in \mathcal{N},\\k \in \mathcal{K}}}\alpha_{n,k} (C_{n,k}-\sum_{j \in \mathcal{M}}x_{n,k}^{j}) +\sum_{\substack{j \in \mathcal{M},\\k \in \mathcal{K}}}\beta_{k}^j (r_{k}^{j}-\sum_{n \in \mathcal{N}}x_{n,k}^{j}) + \sum_{n \in \mathcal{N}} \zeta_{n}  \nonumber \\
	& \hspace*{-1.3in}\quad \bigg(\sum_{\substack{j \in \mathcal{M}_n,\\k \in \mathcal{K}}}\Big(u_n^j(\sum_{m \in \overline{\mathcal{N}}}{x}_{m,k}^j+{x}_{n,k}^j)- u^j(\sum_{m \in \overline{\mathcal{N}}}{x}_{m,k}^j) \Big)+ \quad \Big(\sum_{\substack{l \in \overline{\mathcal{M}},\\k \in \mathcal{K}}}\big( u^l(\sum_{m \in \overline{\mathcal{N}}}{x}_{m,k}^l+{x}_{n,k}^l)-u^l(\sum_{m \in \overline{\mathcal{N}}}{x}_{m,k}^l)-\nonumber\\
	&\hspace*{-1.3in}\quad  D^l({x}_{n,k}^l)\big)\Big)-d_n^0
	\bigg)+\sum_{\substack{n\in \mathcal{{N}},\\l\in \overline{\mathcal{M}},\\k\in\mathcal{K}}}\pi_{n,k}^l \Big(u^l(\sum_{m \in \overline{\mathcal{N}}}{x}_{m,k}^l+{x}_{n,k}^l)-
	u^l(\sum_{m \in \overline{\mathcal{N}}}{x}_{m,k}^l)-  D^l({x}_{n,k}^l)\Big).
	\end{align}
	\hrulefill
\end{figure*}
\section{Distributed algorithm for NBS}\label{sec:dist}
While resource allocation and sharing  using a central algorithm is feasible, it is desirable to develop low overhead distributed algorithms. 
To obtain a distributed algorithm,  we rely on \emph{Duality Theory} and use the \emph{Lagrangian} function   in \eqref{eq:lagrangian}. 
The dual optimization problem is given by 
\begin{align}\label{eq:dualopt}
\min_{\boldsymbol{\alpha},\boldsymbol{\beta},\boldsymbol{\zeta,\gamma,\pi}} D(\boldsymbol{\alpha},\boldsymbol{\beta},\boldsymbol{\zeta,\gamma,\pi})&= \mathcal{L}\big(\mathbf{X}^{*}(\boldsymbol{\alpha,\beta},\boldsymbol{\zeta,\gamma,\pi}),\nonumber \\& \boldsymbol{\alpha,\beta,\zeta,\gamma,\pi} \big), \nonumber \\
\; s.t. \; \boldsymbol{\alpha},\boldsymbol{\beta},\boldsymbol{\zeta,\gamma,\pi}\geq 0.
\end{align}
\begin{align}\label{eq:supLag}
\mathbf{X}^{*}(\boldsymbol{\alpha,\beta},\boldsymbol{\zeta,\gamma,\pi})= \arg \max_{\mathbf{X}} \mathcal{L}(\mathbf{X},\boldsymbol{\alpha,\beta},\boldsymbol{\zeta,\gamma,\pi} ).
\end{align}
The dual problem can be solved iteratively using gradient descent (since the dual problem is a minimization problem) as below:
\begin{align}\label{eq:iterativeDual}
\alpha_{n,k}[t+1] &=  \alpha_{n,k}[t]- \phi_{n,k} \frac{\partial \mathcal{L}(\mathbf{X}, \boldsymbol{\alpha}, \boldsymbol{\beta},\boldsymbol{\zeta,\gamma,\pi}) }{\partial  \alpha_{n,k}}, \nonumber \displaybreak[0] \\
\beta_{k}^j[t+1] &=  \beta_{k}^j[t]- \eta_{k}^j \frac{\partial \mathcal{L}(\mathbf{X}, \boldsymbol{\alpha}, \boldsymbol{\beta},\boldsymbol{\zeta,\gamma,\pi}) }{\partial  \beta_{k}^j}, \nonumber \\
\zeta_{n}[t+1] &=  \zeta_{n}[t]- \omega_{n} \frac{\partial \mathcal{L}(\mathbf{X}, \boldsymbol{\alpha}, \boldsymbol{\beta},\boldsymbol{\zeta,\gamma,\pi}) }{\partial  \zeta_{n}}, \nonumber \displaybreak[0]\\
\gamma_{n,k}^{j}[t+1] &=  \gamma_{n,k}^{j}[t]- \theta_{n,k}^{j} \frac{\partial \mathcal{L}(\mathbf{X}, \boldsymbol{\alpha}, \boldsymbol{\beta},\boldsymbol{\zeta,\gamma,\pi}) }{\partial  \gamma_{n,k}^{j}},\nonumber \\
\pi_{n,k}^{l}[t+1]&=\pi_{n,k}^{l}[t]-\psi_{n,k}^l \frac{\partial \mathcal{L}(\mathbf{X}, \boldsymbol{\alpha}, \boldsymbol{\beta},\boldsymbol{\zeta,\gamma,\pi}) }{\partial  \pi_{n,k}^{l}},
\end{align}
where $\phi_{n,k}$, $\eta_{k}^j$,  $\omega_{n}$, $\theta_{n,k}^{j}$ and $\psi_{n,k}^l$ are positive step sizes. Furthermore, the gradients in \eqref{eq:iterativeDual} are given below.
\begin{align}\label{eq:gradforLagrangieMultipliers}
\frac{\partial \mathcal{L}(\mathbf{X}, \boldsymbol{\alpha}, \boldsymbol{\beta}, \boldsymbol{\zeta,\gamma,\pi}) }{\partial  \alpha_{n,k}}&= (C_{n,k}-\sum_{j \in \mathcal{M}}x_{n,k}^{j}),\nonumber \\
\frac{\partial \mathcal{L}(\mathbf{X}, \boldsymbol{\alpha}, \boldsymbol{\beta}, \boldsymbol{\zeta,\gamma,\pi}) }{\partial  \beta_{k}^j} &= (r_{k}^{j}-\sum_{n \in \mathcal{N}}x_{n,k}^{j}),\nonumber \\
\frac{\partial \mathcal{L}(\mathbf{X}, \boldsymbol{\alpha}, \boldsymbol{\beta},\boldsymbol{\zeta,\gamma,\pi}) }{\partial  \zeta_{n}} &=\bigg(\sum_{\substack{j \in \mathcal{M}_n,\\k \in \mathcal{K}}}\Big(u^j(\sum_{m \in \overline{\mathcal{N}}}{x}_{m,k}^j+{x}_{n,k}^j)-  \nonumber \\
&\hspace*{-1.2in}u^j(\sum_{m \in \overline{\mathcal{N}}}{x}_{m,k}^j) \Big)+\Big(\sum_{\substack{l \in \overline{\mathcal{M}},\\k \in \mathcal{K}}}\big( u^l(\sum_{m \in \overline{\mathcal{N}}}{x}_{m,k}^l+{x}_{n,k}^l)-\nonumber\\
&\hspace*{-1.2in}u^l(\sum_{m \in \overline{\mathcal{N}}}{x}_{m,k}^l)- D^l({x}_{n,k}^l)\big)\Big)-d_n^0
\bigg)
, \nonumber \\
\frac{\partial \mathcal{L}(\mathbf{X}, \boldsymbol{\alpha}, \boldsymbol{\beta},\boldsymbol{\zeta,\gamma,\pi}) }{\partial  \gamma_{n,k}^{j}} &= x_{n,k}^{j},\nonumber\\
\frac{\partial \mathcal{L}(\mathbf{X}, \boldsymbol{\alpha}, \boldsymbol{\beta},\boldsymbol{\zeta,\gamma,\pi}) }{\partial  \pi_{n,k}^{l}} &=\Big(u^l(\sum_{m \in \overline{\mathcal{N}}}{x}_{m,k}^l+{x}_{n,k}^l)-\nonumber \\
&\hspace*{-1.2in}
u^l(\sum_{m \in \overline{\mathcal{N}}}{x}_{m,k}^l)-  D^l({x}_{n,k}^l)\Big).
\end{align}
Unless certain  conditions such as Slater's constraint qualification  are satisfied, \emph{strong duality}\footnote{Strong duality implies that there is no duality gap between the primal and dual problem.} is not guaranteed to hold for our primal and dual problems \cite{boyd2004convex}. However,  we rely on the following theorems to show that strong duality holds.
\begin{theorem}[Sufficient Condition\cite{tychogiorgos2013non}]\label{thm:georgesufficient}
	If the price based function $\mathbf{X}^*(\boldsymbol{\alpha,\beta},\boldsymbol{\zeta,\gamma,\pi})$ is continuous at one or more of the optimal Lagrange multipliers,  the iterative algorithm consisting of  \eqref{eq:supLag} and \eqref{eq:iterativeDual} will converge to the global optimal solution. 
\end{theorem}
\begin{theorem}[Necessary Condition\cite{tychogiorgos2013non}]\label{thm:georgenecessary}
	The condition in Theorem \ref{thm:georgesufficient} is also necessary if at least one of the constraints in \eqref{eq:objmain2} is active (binding) at the optimal solution. 
\end{theorem}
\begin{lemma}\label{lem:constraintsactive}
	At the optimal solution, the optimal Lagrange multiplier vector corresponding to the capacity constraint is non-zero, i.e., $\boldsymbol{\alpha_k^{*}}>0$. 
\end{lemma}
\begin{proof}
	At the optimal point, all  resources are fully utilized,  i.e., capacity constraints are active. From \emph{complementary slackness}\cite{boyd2004convex}, we know that 
	\begin{align}\label{eq:compslackness}
	\alpha^*_{n,k}\bigg(\sum_{j\in \mathcal{M}} x_{n,k}^{*j}- C_{n,k}\bigg)=0, \forall n \in \mathcal{N}. 
	\end{align}
	Since $\sum_{j\in \mathcal{M}} x_{n,k}^{*j}= C_{n,k}$, i.e., the constraint is active, which implies that $\boldsymbol{\alpha_k^{*}} >0$.
\end{proof}
\begin{theorem}\label{thm:duality}
For any concave injective utility function for which the inverse of the first derivative exists, the iterative distributed algorithm consisting of  \eqref{eq:supLag} and \eqref{eq:iterativeDual}  converges to the global optimal solution.
\end{theorem}
\begin{proof}
	We show that the price function obtained by $\frac{\partial \mathcal{L}(\mathbf{X}, \boldsymbol{\alpha}, \boldsymbol{\beta},\boldsymbol{\zeta,\gamma,\pi})}{\partial x_{n,k}^{j}}=0$ is continuous at one or more of the optimal Lagrange multipliers for \eqref{eq:lagrangian}.  
	
	Let  $u^{'-j}(\boldsymbol{.})$ 
	and $f^{'-l}(\boldsymbol{.})$ 
	represent   
	the inverses of the first derivative of $u^j$ and $\big(u^l(\sum_{m\in\overline{\mathcal{N}}} x_{m,k}^l+x_{n,k}^l)-u^l(\sum_{m\in\overline{\mathcal{N}}} x_{m,k}^l) -D^l(x_{n,k}^l)\big)$, respectively. 
	$\alpha_{n,k}$,  $\beta_{k}^j$, $\zeta_{n}$,  $\gamma_{n,k}^{j}$ and $\pi_{n,k}^l$ are the \emph{Lagrange} multipliers whereas:   
	\begin{align}
	\Delta_{n}&=\sum_{\substack{j \in \mathcal{M}_n,\\k \in \mathcal{K}}}\Big(u^j(\sum_{m \in \overline{\mathcal{N}}}{x}_{m,k}^j+{x}_{n,k}^j)- u^j(\sum_{m \in \overline{\mathcal{N}}}{x}_{m,k}^j) \Big)+ \nonumber \\
	&\Big(\sum_{\substack{l \in \overline{\mathcal{M}},\\k \in \mathcal{K}}}\big( u^l(\sum_{m \in \overline{\mathcal{N}}}{x}_{m,k}^l+{x}_{n,k}^l)-u^l(\sum_{m \in \overline{\mathcal{N}}}{x}_{m,k}^l)- \nonumber\\
	&D^l({x}_{n,k}^l)\big)\Big)-d_n^0
	\end{align}
	We consider two different cases
	
	\subsubsection{	\textbf{$j \in \mathcal{M}_n$:}$\frac{\partial \mathcal{L}}{\partial x_{n,k}^{j}}=0$}
		\begin{align}
		\label{eq:gradLagObj22}
		&u^{'j}(\sum_{m \in \overline{\mathcal{N}}}{x}_{m,k}^j+x_{n,k}^{j})(1+\Delta_{n}^{}\zeta_{n})-\Delta_{n}^{}(\alpha_{n,k}+\nonumber \\
		&\beta_{k}^j-\gamma_{n,k}^{j})=0,\nonumber \\	
		&\implies u^{'j}(\sum_{m \in \overline{\mathcal{N}}}{x}_{m,k}^j+x_{n,k}^{j})=\frac{\Delta_{n}^{}(\alpha_{n,k}+\beta_{k}^j- \gamma_{n,k}^{j})}{1+\Delta_{n}^{}\zeta_{n}}, \nonumber \\
		& x_{n,k}^{j}=u^{'-j}\Bigg(\frac{\Delta_{n}^{}(\alpha_{n,k}+\beta_{k}^j- \gamma_{n,k}^{j})}{1+\Delta_{n}^{}\zeta_{n}}\Bigg)-\sum_{m \in \overline{\mathcal{N}}}{x}_{m,k}^j. 
		\end{align}
		 We prove that \eqref{eq:gradLagObj22} is continuous at the optimal point by showing that the numerator $\Delta_{n}^{}(\alpha_{n,k}+\beta_{k}^j- \gamma_{n,k}^{j})$ and denominator $1+\Delta_{n}^{}\zeta_{n} $ are positive\footnote{Positive numerator and denominator are required if $u^{'-j}$ is $\log$. For other cases, it will suffice to prove that the denominator is non-zero. }. 	At the optimal point, the denominator $(1+\Delta_{n}^{*}\zeta^*_{n})$ is positive as $\Delta_{n}^{*}>0$ (from Section \ref{sec:assumption}) and $\zeta^*_{n}=0$ (from complementary slackness).	Similarly, in the numerator, 
		 $(\alpha^*_{n,k}+\beta_{k}^{*j}-\gamma_{n,k}^{*j})>0$, as  $\alpha_{n,k}^{*}>0$ (from Lemma \ref{lem:constraintsactive}) and  $\alpha_{n,k}^{*}+\beta_{k}^{*j}> \gamma_{n,k}^{*j}$ (from \emph{sensitivity analysis} \cite{boyd2004convex}). Hence, \eqref{eq:gradLagObj22} is continuous at the optimal point.
\\		
		\subsubsection{\textbf{$l \in \{\mathcal{M}\backslash\mathcal{M}_n\}$:} $\frac{\partial \mathcal{L}}{\partial x_{n,k}^{l}}=0$}
		\begin{align}
		\label{eq:gradLagObjcase22}
		&f^{'l}(\sum_{m \in \overline{\mathcal{N}}}{x}_{m,k}^l+x_{n,k}^{l})(1+\Delta_{n}^{}\zeta_{n}+\Delta_{n}^{}\pi_{n,k}^l)-\Delta_{n}^{}(\alpha_{n,k}+\nonumber \\
		&\beta_{k}^l-\gamma_{n,k}^{l})=0,\nonumber \\
		&\implies  f^{'l}(\sum_{m \in \overline{\mathcal{N}}}{x}_{m,k}^l+x_{n,k}^{l})=\frac{\Delta_{n}^{}(\alpha_{n,k}+\beta_{k}^l-\gamma_{n,k}^{l})}{1+\Delta_{n}^{}\zeta_{n}+\Delta_{n}^{}\pi_{n,k}^j}. \nonumber \\
		&  x_{n,k}^{l}=f^{'-l}\Bigg(\frac{\Delta_{n}^{}(\alpha_{n,k}+\beta_{k}^l-\gamma_{n,k}^{l})}{1+\Delta_{n}^{}\zeta_{n}+\Delta_{n}^{}\pi_{n,k}^l}\Bigg)-\sum_{m \in \overline{\mathcal{N}}}{x}_{m,k}^l. \nonumber \\
		\end{align}
	The continuity at optimal point can be established using arguments similar  to that for  $j \in \mathcal{{M}}_n$ case. The proof of zero-duality gap follows from Theorem \ref{thm:georgesufficient}.  Hence, strong duality holds and our proposed distributed algorithm converges to the NBS. 
\end{proof}

For calculating $x_{n,k}^{j}$ at time $t$, we use $\Delta_{n}^{}$ calculated at $t-1$. 
Algorithm \ref{algo:primal-dual-alg} is a distributed algorithm that provides the global optimal solution to \eqref{eq:nashProbNplayermod}.
\paragraph{Protocol for Distributed Execution}
All ESPs first broadcast information regarding resource requests from their native applications. Then any randomly chosen ESP starts allocating resources as specified in Step $2$ of Algorithm \ref{algo:primal-dual-alg} and updates the corresponding Lagrangian multipliers. This ESP then  passes on the information about its allocation  and updated request matrix to the next ESP (that has not yet allocated resources in the current round\footnote{A single round consists of all the ESPs executing Step $2$ of Algorithm \ref{algo:primal-dual-alg} once.}) that repeats the  same procedure until all the ESPs allocate their resources.    This process continues until the first order conditions given in \cite{ktlagrangian} are met. 

\begin{algorithm}
	\begin{algorithmic}[]
		\State \textbf{Input}: $\forall\; \boldsymbol{\alpha_0,\; \beta_0,\zeta_0,\gamma_0,\pi_0}$, $C, R$, vector of utility function of all ESPs $\boldsymbol{u}$  and $\mathbf{X}_0$ 
		\State \textbf{Output}: The optimal resource allocation $\mathbf{X} $ and payoffs of all ESPs
		\State \textbf{Step $0$:}  $t=0$, $\boldsymbol \alpha[t]\leftarrow\boldsymbol \alpha_0$, $\boldsymbol\beta[t]\leftarrow\boldsymbol\beta_0$, $\boldsymbol\zeta[t]\leftarrow\boldsymbol\zeta_0$, $\boldsymbol\gamma[t]\leftarrow\boldsymbol\gamma_0$,$\boldsymbol\pi[t]\leftarrow\boldsymbol\pi_0$$, \mathbf{X}[t]\leftarrow \mathbf{X}_0$
			\State \textbf{Step $1$:}
		\For{\texttt{$n \in \mathcal{N}$}}
		\State $ {d_n^0} \leftarrow$ \texttt{Objective function at optimal point in Equation \eqref{eq:opt_single}}
		\EndFor
		\State \textbf{Step 2: $t\geq 1$} 
		\While{{First order conditions\cite{ktlagrangian}$\neq$true }}
		\State Compute $ x_{n,k}^{j}[t+1]$ for $j\in\mathcal{M},$ $k\in\mathcal{K}$ and $n\in\mathcal{N}$ through \eqref{eq:gradLagObj22} and \eqref{eq:gradLagObjcase22};
		\State Compute $\boldsymbol\alpha[t+1]$, $\boldsymbol\beta[t+1]$, $\boldsymbol\zeta[t+1]$,  $\boldsymbol\gamma[t+1]$ and $\boldsymbol\pi[t+1]$ through~(\ref{eq:iterativeDual}) given $\mathbf{X}[t+1],$ $\boldsymbol\alpha[t]$,  $\boldsymbol \beta[t]$, $\boldsymbol \zeta[t]$,  $\boldsymbol\gamma[t]$ and $\boldsymbol\pi[t]$
		\EndWhile
	\end{algorithmic}
	\caption{Distributed Algorithm for NBS}
	\label{algo:primal-dual-alg}
\end{algorithm}

\section{Simulation Results}\label{sec:exp_results}

We evaluate the performance of the proposed NBS framework for resource sharing and allocation across several  settings  in Table \ref{tab:settings}. 
Each SP has three different  resources, i.e., storage, communication and computation. The model can be extended to include other  resources/parameters. 
We study the proposed framework using both synthetic and real-world data traces \cite{shen2015statistical,kohne2014federatedcloudsim,kohne2016evaluation}. 
For the study with synthetic and real-world data traces,  request matrices and capacity vectors, $\forall n \in \mathcal{N}$ are randomly generated for each setting. 
To show the advantage of resource sharing for both synthetic and   data trace,  we set large resource capacities at certain SPs so that they can improve their utilities by sharing  available resources with other SPs for meeting the demand of  latter SPs. 

Simulations were run in \texttt{Matlab R2019a} on a \texttt{Core-i7} processor with \texttt{16 GB RAM}. To solve the optimization problems in \eqref{eq:opt_single} and \eqref{eq:objmain2}, we use the \texttt{OPTI-toolbox}\cite{currie2012opti}. 
We  evaluate our proposed algorithms from the perspective of  service providers that are interested in maximizing their utilities and  evaluate the impact of our framework on the applications.	
We define  application \emph{request satisfaction} (RS)  as the ratio of  allocated  resources to   requested resources. 
Mathematically, average RS for an SP $n$ in the resource sharing case  is defined as: 
\begin{align}\label{eq:requestsatisfaction}
\centering 
RS_{n}= \frac{\sum_{j \in \mathcal{M}_n}\sum_{k \in \mathcal{K}}\big(\frac{\sum_{m \in \mathcal{{N}}}x_{m,k}^{j}}{r_{n,k}^{j}}\big)}{M_nK} \times 100, 
\end{align} 
The utility  and communication cost functions used in simulation are given in \eqref{eq:utility}.

\begin{align}\label{eq:utility}
\centering
u^j(x_{n,k}^j) & = 1-e^{-(x_{n,k}^j-r_{n,k}^j+\delta)}, \nonumber \\
D^j(x_{n,k}^j) &= \frac{x_{n,k}^j}{w},
\end{align}
Here $\delta$ is set to $1$ and weight $w$ is randomly chosen.

\begin{table}[]
	\centering
	\caption{Simulation network settings for NBS based resource sharing
		framework.}
	\begin{tabular}{|l|l|}
		\hline
		\textbf{Setting}   & \textbf{Parameters} \\ \hline
		\textbf{1} & $N=3, M_n=3,\forall n \in \mathcal{N}, K=3$                    \\ \hline
		\textbf{2} & $N=3, M_n=20,\forall n \in \mathcal{N}, K=3$                    \\ \hline
		\textbf{3 (data traces)} & $N=3, M_n=20,\forall n \in \mathcal{N}, K=3$                    \\ \hline
		\textbf{4} & $N=6, M_n=6,\forall n \in \mathcal{N}, K=3$                    \\ \hline
		\textbf{5} & $N=6, M_n=20,\forall n \in \mathcal{N}, K=3$                    \\ \hline
	\end{tabular}
	\label{tab:settings}
\end{table}

\begin{figure*}
	
	\centering
	\includegraphics[width=1.0\textwidth]{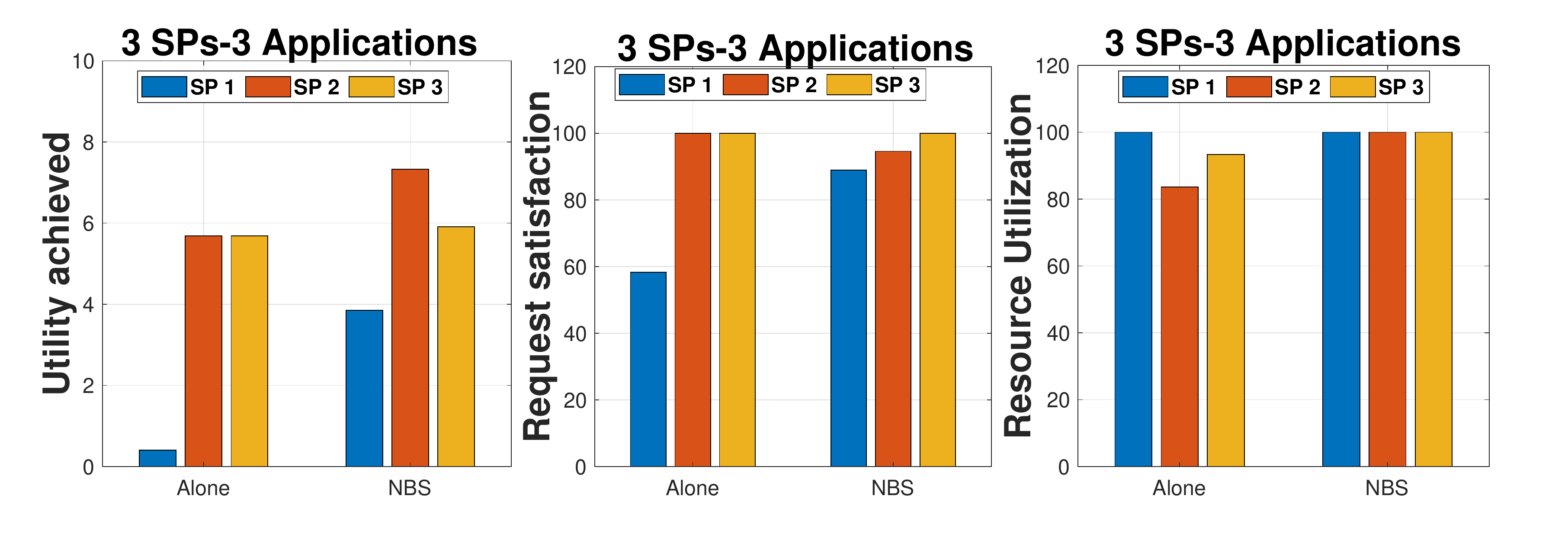}
	\caption{Utility, average request satisfaction and average resource utilization for Setting $1$ when SPs are working alone and using our proposed NBS framework.}
	\protect\label{fig:withNBS1}
\end{figure*}

\begin{figure*}
	\centering
	\includegraphics[width=1.0\textwidth]{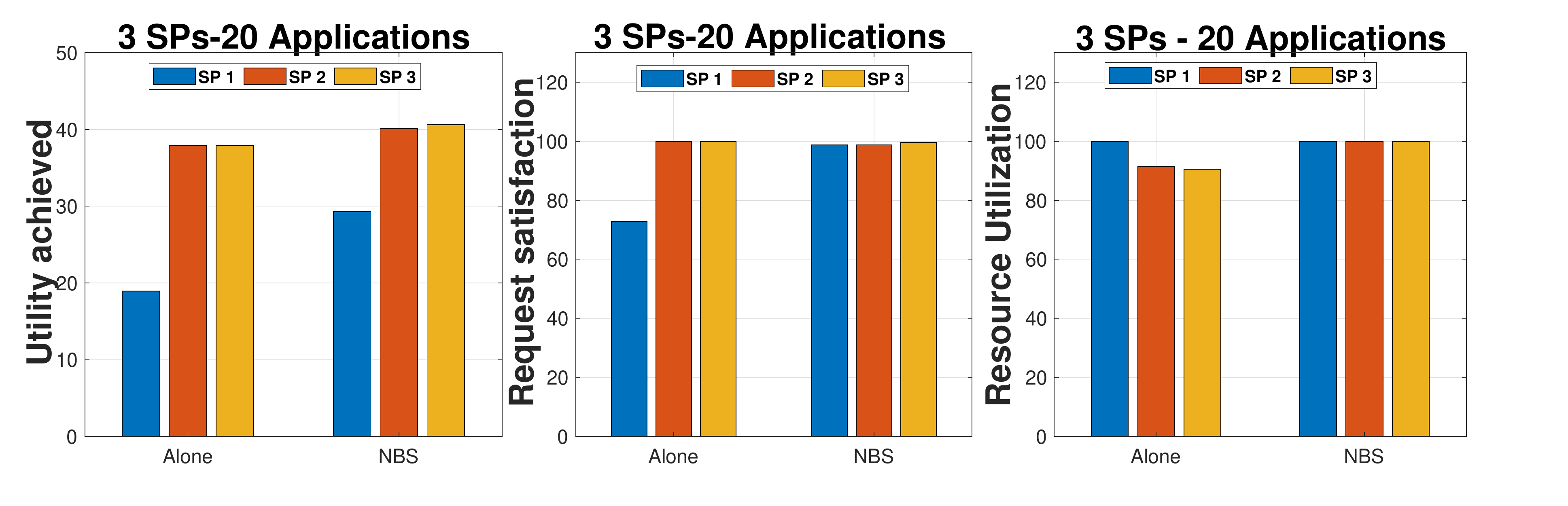}
	\caption{Utility, average request satisfaction and average resource utilization for Setting $2$ when SPs are working alone and using our proposed NBS framework.}
	\protect\label{fig:withNBS2}
\end{figure*}

\subsection{Simulations results for synthetic data}
\begin{figure*}
	\centering
	\includegraphics[width=1.0\textwidth]{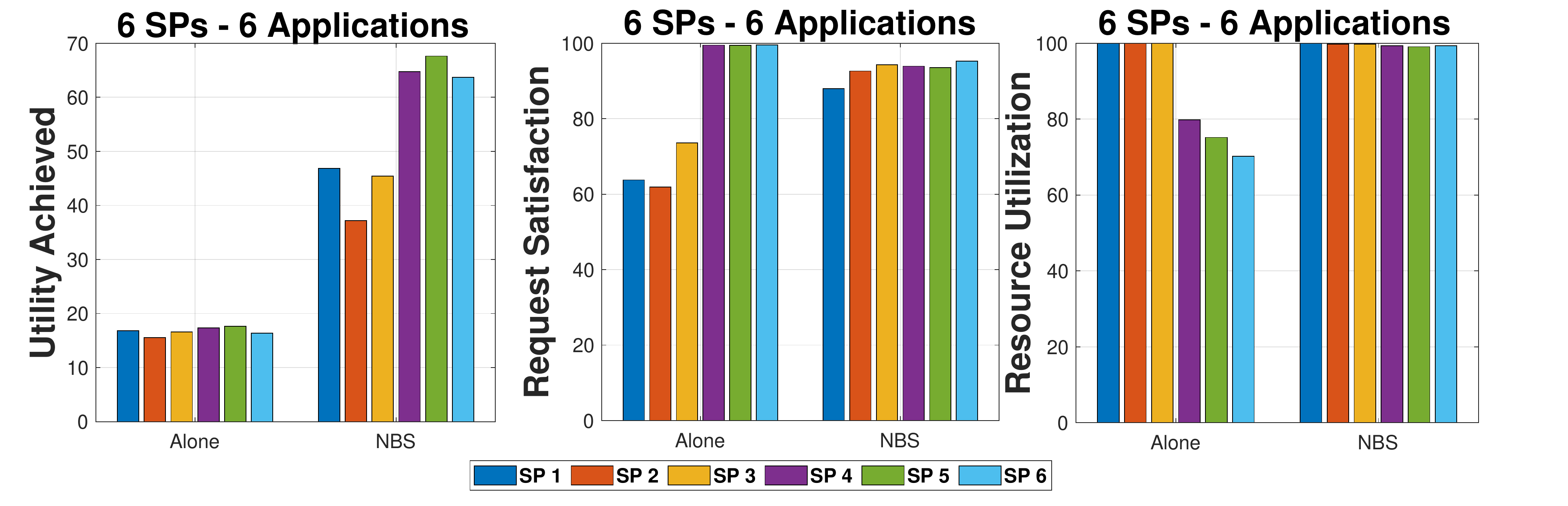}
	\caption{Utility, average request satisfaction and average resource utilization for Setting $4$ when SPs are working alone and using our proposed NBS framework.}
	\protect\label{fig:withNBS66}
\end{figure*}
\begin{figure*}
	\centering
	\includegraphics[width=1.00\textwidth]{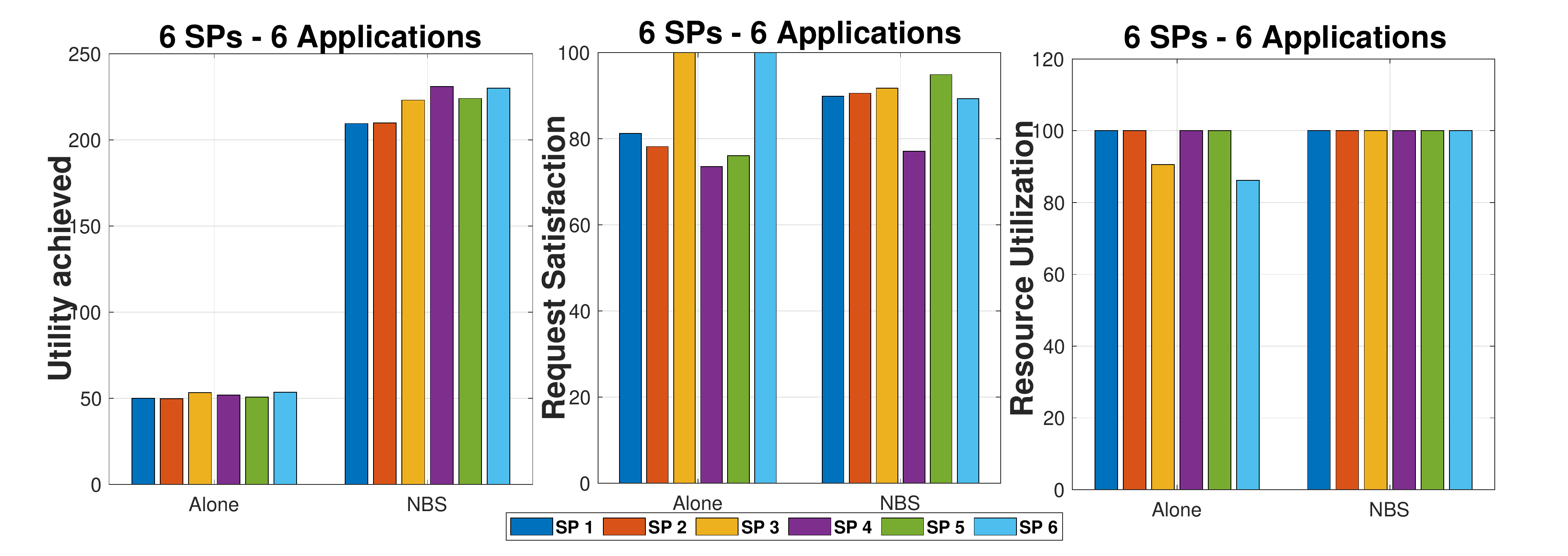}
	\caption{Utility, average request satisfaction and average resource utilization for Setting $5$ when SPs are working alone and using our proposed NBS framework.}
	\protect\label{fig:withNBS7}
\end{figure*}
To highlight the efficacy of our framework, we compare its performance with a setting where SPs  work alone (i.e., no resource sharing among edge SPs). 
In particular, we compare SP utility, the average resource utilization (averaged across all $k$ resources) and the average request satisfaction (averaged across requests of all SPs applications) in Figures \ref{fig:withNBS1}, \ref{fig:withNBS2}, \ref{fig:withNBS66}, \ref{fig:withNBS7} for settings 1, 2, 4 and 5, respectively. 
For the 3 SP settings, when SP $1$  works alone, it has a resource deficit (evident from $100\%$ average resource utilization and  average request satisfaction of less than 60\%  and 80\% in Figures \ref{fig:withNBS1} and \ref{fig:withNBS2}, respectively) whereas SPs $2$ and $3$ have  resource surpluses as indicated by less than $100\%$ resource utilization and $100\%$ request satisfaction. The resource deficit results in a lower utility and request satisfaction for SP $1$. 

On the other hand,  both SPs $2$ and $3$   achieve  higher utilities by satisfying all their applications when working alone. However, by using our resource sharing framework, the utilities of all SPs improve as the framework provides optimal resource sharing.  
For the case with  three applications,  average request satisfaction improves from $86.11\%$ (working alone) to $94.5\%$ (resource sharing) whereas it improves from  $90.9\%$ to $99\%$ for the $20$ application case. It is worth noting  that request satisfaction for native applications of SPs $2$ and $3$ reduce as these SPs allocate their resources to applications of SP $1$ for a higher utility. 
Furthermore,   resource utilizations also increase for the SPs with resource surpluses as they share their resources with the SP with a resource deficit.  Similar results are obtained for settings $4$ and $5$ given in Table \ref{tab:settings}. For setting $4$, the request satisfaction improves from 82.96\% to 92.94\% using our proposed NBS framework, whereas request satisfaction improves in setting $5$ from 84.85\% to 88.94\% using our proposed resource sharing framework. Table \ref{tab:summreqres} summarizes request satisfactions and resource utilization in different settings using our framework and when working alone. 
\begin{figure*}
	\centering
	\includegraphics[width=1.00\textwidth]{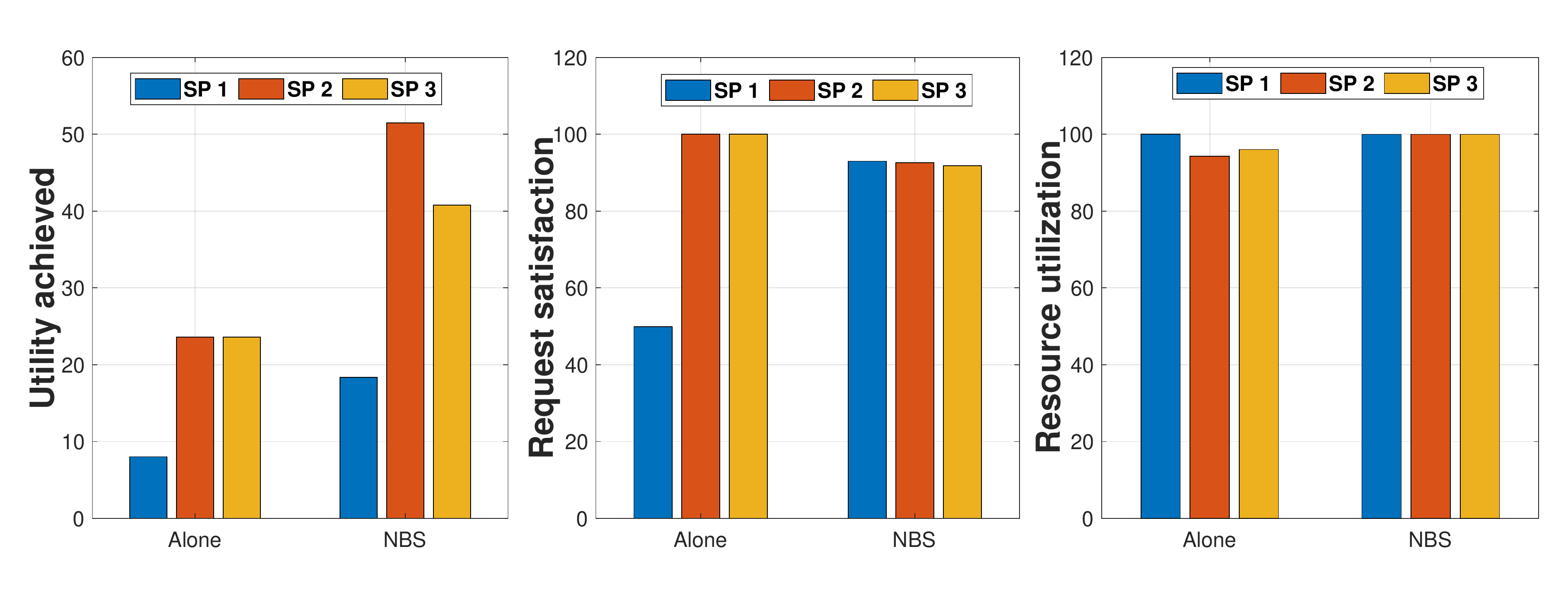}
	\caption{Utility, average request satisfaction and average resource utilization for Setting $3$ when SPs are working alone and using our proposed NBS framework.}
	\protect\label{fig:withNBS6}
	
\end{figure*}
\begin{table}[]
	\centering 
	\caption{Summary of average request satisfaction and resource utilization in different settings with and without the proposed NBS sharing framework.}
	\begin{tabular}{|l|l|l|l|l|}
		\hline
		\multirow{2}{*}{\textbf{Setting}} & \multicolumn{2}{l|}{\textbf{Request Satisfaction(\%)}} & \multicolumn{2}{l|}{\textbf{Resource Utilization(\%)}} \\ \cline{2-5} 
		& \textbf{Alone}            & \textbf{NBS}           & \textbf{Alone}            & \textbf{NBS}           \\ \hline
		\textbf{1}                        & 91.39                     & 96.90                & 93.37                     & 100                    \\ \hline
		\textbf{2}                        & 90.93                     & 99.04                & 94.01                   & 100                    \\ \hline
		\textbf{3}                        & 83.30                   & 92.44               & 96.76                   & 100                    \\ \hline
		\textbf{4}                        & 82.96                     & 92.94                  & 87.54                     & 100                    \\ \hline
		\textbf{5}                        & 84.85                     & 88.94                  & 96.13                     & 100                    \\ \hline
	\end{tabular}
	\label{tab:summreqres}
\end{table}

\subsection{Results for the data traces}
We use trace files from \emph{fastStorage}, \emph{Rnd} \cite{shen2015statistical}, and \emph{materna} \cite{kohne2014federatedcloudsim,kohne2016evaluation}. 
We simulate a setting with three different SPs and randomly extract the normalized resource request information related to the number of CPU cores, the amount of CPU and memory (RAM) for 20 different resource requests  from  \emph{fastStorage}, \emph{Rnd} and \emph{materna} dataset. Since the datasets do not provide the capacities of these service providers, we assign capacities in such a way that fastStorage serves as the SP with resource deficit while the other two have resource surplus. Figure \ref{fig:withNBS6} shows SP utilities,  request satisfaction and resource utilization based on  the  data traces. It is evident that SP utilities  improve  and resource utilizations  increase to satisfy more applications by use of the proposed NBS framework. The average request satisfaction also increases from $83.3\%$ to $92.05\%$.  

\subsection{Measure of Fairness}

NBS is known for its fairness property \cite{han2012game}. In this section, we show the fairness of our proposed NBS based resource sharing framework for different settings given in Table \ref{tab:settings}. In particular, to measure the degree of fairness of the proposed sharing framework, we calculate  \emph{Jain}'s index\footnote{$\frac{1}{|N|}\leq$Jain's index$\leq 1$ where $1$ is the highest value of fairness.} \cite{jain1999throughput,jain1984aquantitative}. Figure \ref{fig:nbsfairness} shows Jain's index for different settings given in Table \ref{tab:settings}. The value of Jain's index is larger than 0.95 in all the settings.  
Especially for scenarios with a large number of applications, these results reveal  that our framework enables  fair sharing and allocation of resources among SPs, as one would expect from the product-based fairness as offered by the NBS.

\begin{figure}
	\includegraphics[width=0.54\textwidth]{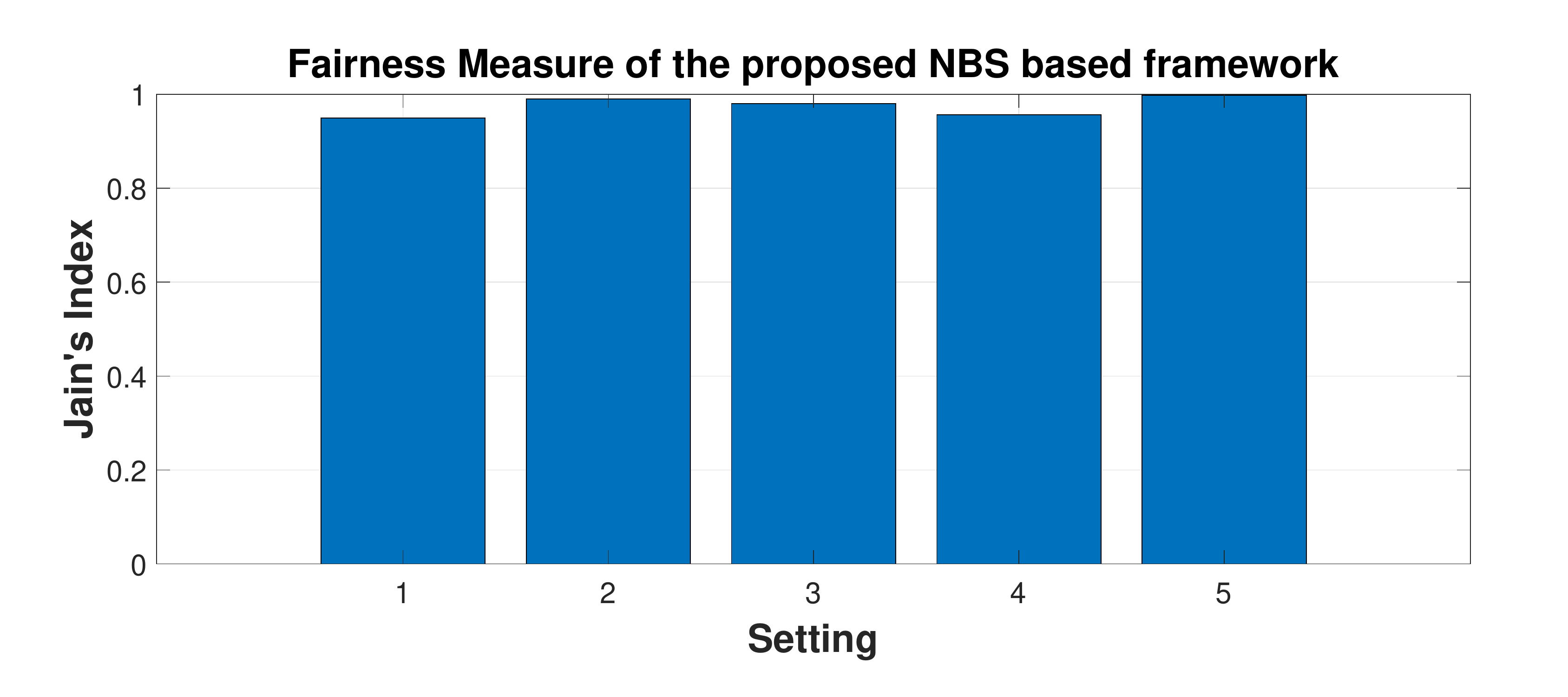}
	\caption{Jain's Index in different settings using our proposed NBS based resource sharing framework.}
	\protect\label{fig:nbsfairness}
\end{figure}
\section{Related Work}\label{sec:related}
\subsection{Resource allocation in Edge Computing}
Jiao et al. \cite{jiao2018social} propose an auction-based resource allocation scheme for edge service providers that provide resources for blockchains. The proposed mechanism maximizes the social welfare and guarantees truthfulness, computational efficiency and individual rationality. He et al. \cite{he2018s} analyze the optimal provisioning of shareable and non-shareable edge resources to different applications.    
Plachy et al. \cite{plachy2016dynamic} consider the mobility problem in mobile edge clouds and propose a novel algorithm for  selecting communication path and VM placement. The proposed approach relies on predicting user movement that helps in VM placement and accordingly selecting the communication path. Nishio et al. \cite{nishio2013service} propose a resource sharing architecture for mobile clouds that relies on service-oriented utility functions and is closest to our work. However, they primarily consider service latency and rely on a centralized framework whereas we present a distributed framework that requires a specific utility function that can be used to model different metrics such as latency, delay, and numerous other objectives. In \cite{zafari2020let}, we modeled resource sharing among mobile edge clouds as a cooperative game. However, the MOO problem in \cite{zafari2020let} is non-convex  that is hard to solve compared to the convex problem solved here. Furthermore, the cooperative game based framework lacks the fairness guaranteed by NBS. 
 Furthermore, in contrast with \cite{jiao2018social,plachy2016dynamic} and \cite{he2018s} that primarily consider resource allocation, our proposed framework deals with resource sharing among ESPs with different objectives.

\begin{table*}[]
	\centering
	\caption{Comparison of our NBS based resource sharing and allocation framework with other NBS based solutions.}
	\begin{tabular}{|c|p{1.0cm}|p{2.1cm}|c|p{5.5cm}|}
		\hline
		\multirow{2}{*}{\textbf{Reference}} & \multirow{2}{*}{\textbf{Objective}} & \multirow{2}{*}{\textbf{Resource Sharing}} & \multicolumn{2}{c|}{\textbf{Distributed Algorithm}}             \\ \cline{4-5} 
		&                                     &                                            & \textbf{Proposed?}      & \multicolumn{1}{c|}{\textbf{Technique}} \\ \hline
		\multicolumn{1}{|l|}{\cite{yaiche2000game}}              & \multicolumn{1}{l|}{Fair Bandwidth allocation}               & \multicolumn{1}{c|}{$\times$}                      & \multicolumn{1}{c|}{$\checkmark$} & Gradient Projection                                   \\ \hline		
		
		\multicolumn{1}{|l|}{\cite{xu2012general}}              & \multicolumn{1}{l|}{Fairness}               & \multicolumn{1}{c|}{$\times$}                      & \multicolumn{1}{c|}{$\checkmark$} & Dual decomposition and sub-gradient method are used. 
		\\ \hline
		
		\multicolumn{1}{|l|}{\cite{hassan2014virtual}}              & \multicolumn{1}{l|}{Cost and Resource Utilization}               & \multicolumn{1}{c|}{$\times$}                      & \multicolumn{1}{c|}{$\times$} & N/A                                     \\ \hline
		
		\multicolumn{1}{|l|}{\cite{he2013toward}}              & \multicolumn{1}{l|}{Cost and User Experience}               & \multicolumn{1}{c|}{$\times$}                      & \multicolumn{1}{c|}{$\times$} & N/A                                     \\ \hline
		
		\multicolumn{1}{|l|}{\cite{feng2012bargaining}}              & \multicolumn{1}{l|}{Resource utilization}               & \multicolumn{1}{c|}{$\times$}                      & \multicolumn{1}{c|}{$\times$} & N/A                                  \\ \hline

		\multicolumn{1}{|l|}{\cite{guo2013cooperative}}              & \multicolumn{1}{l|}{Bandwidth allocation}               & \multicolumn{1}{c|}{$\times$}                      & \multicolumn{1}{c|}{$\checkmark$} & Gradient Projection                                   \\ \hline		
		
		\multicolumn{1}{|l|}{Our Approach}              & \multicolumn{1}{l|}{Service provider utility and user satisfaction}               & \multicolumn{1}{c|}{$\checkmark$}                      & \multicolumn{1}{c|}{$\checkmark$} & Gradient descent based algorithm that works for the class of utilities described in section \ref{sec:assumption}.                                   \\ \hline			
		
	\end{tabular}
	\label{tab:comparison}
\end{table*}
\subsection{NBS based Resource Allocation}
Yaiche et al. \cite{yaiche2000game} use NBS to allocate bandwidth for elastic services in  high speed networks. 
Xu et al. \cite{xu2012general} consider the fairness criteria when allocating resources to different cloud users and make use of the  NBS to guarantee fairness. Using dual composition and sub-gradient method, the authors also develop a distributed algorithm. Hassan et al.  \cite{hassan2014virtual} propose an NBS-based model for cost-effective and dynamic VM allocation with multimedia traffic, and show that it can reduce the cost of running different servers along with maximizing resource utilization and satisfying the QoS requirements. He et al. \cite{he2013toward}   study the optimal deployment of content in a cloud assisted video distribution system. 
Feng et al. \cite{feng2012bargaining} use NBS  for Virtual Machine (VM) migration to maximize the resource utilization in a video streaming data center. 
\par 
In contrast with  \cite{yaiche2000game, feng2012bargaining, guo2013cooperative,he2013toward,hassan2014virtual},  to the best of our knowledge, our  framework is first of its kind that uses NBS for  resource sharing among ESPs  with different utilities (objective functions). We show that resource sharing can improve  utilities  of ESPs and enhance application satisfaction. 
Furthermore, for a particular class of utilities, we have proved that distributed algorithm exists for obtaining NBS for the resource sharing and allocation problem. Table \ref{tab:comparison} summarizes some other solutions proposed in the literature  that use NBS for resource allocation in different systems.   Other distributed algorithms proposed in literature either rely on dual decomposition \cite{xu2012general} or gradient projection \cite{yaiche2000game,guo2013cooperative}. However, most of the functions are not dual decomposable and gradient projection is a computationally expensive approach \cite{bertsekas1997nonlinear}, whereas gradient descent is widely used particularly  in machine and deep learning.

\section{Conclusions}\label{sec:conclusion}
The focus in this paper is to optimally utilize available resources for satisfying a larger number of edge applications and improving the utility of edge service providers. We have shown that although ESPs may have different utilities, they should share
resources to improve their utilities and enhance application request satisfaction. Resource sharing among ESPs has been formulated as a bargaining problem and a resource-sharing framework using Nash Bargaining Solution (NBS) has been proposed, which has also been shown to be beneficial for the ESPs. Since a centralized solution for obtaining the NBS may not always be desirable, the strong duality property has been proved, which has enabled us to develop a distributed algorithm for the NBS. Using synthetic and real-world  data traces, we have demonstrated the effectiveness of the proposed framework.  In particular, our results confirm that ESPs with resource
deficits and surpluses can improve their utilities as well as application satisfaction by sharing resources.



  \section*{Acknowledgments}

This work was supported by the U.S. Army Research Laboratory and the U.K. Ministry of Defence under Agreement Number W911NF-16-3-0001. The views and conclusions contained in this document are those of the authors and should not be interpreted as representing the official policies, either expressed or implied, of the U.S. Army Research Laboratory, the U.S. Government, the U.K. Ministry of Defence or the U.K. Government. The U.S. and U.K. Governments are authorized to reproduce and distribute reprints for Government purposes notwithstanding any copy-right notation hereon.  Faheem Zafari also acknowledges the financial support by EPSRC Centre for Doctoral Training in High Performance Embedded and Distributed Systems (HiPEDS, Grant Reference EP/L016796/1), and Department of Electrical and Electronics Engineering, Imperial College London.

\bibliographystyle{IEEETran}
\bibliography{refs}  

\end{document}